\documentclass[twocolumn,article,10pt]{IEEEtran}
\usepackage{epsf}             %incluir figuras encapsulated postscript
\usepackage{amssymb}
\usepackage{lipsum}
\usepackage[latin1]{inputenc}
\usepackage{amssymb,amstext,amsmath,amsthm}
\usepackage{array}
\usepackage[ruled,vlined]{algorithm2e}
\usepackage{verbatim}
\usepackage{float}
\usepackage[dvips]{graphicx}
\usepackage{tikz}
\usepackage{fancyhdr}
\usetikzlibrary{arrows,automata}
\usepackage{tabularx}
\usepackage{cite}
\usepackage[english]{babel}
\usepackage{acronym}
\usepackage{epstopdf}
\newtheorem{theorem}{Theorem}

\begin{document}

\title{MIMO Mutli-Cell Processing: Optimal Precoding and Power Allocation}

\author{Samah A. M. Ghanem,~\IEEEmembership{Senior Member,~IEEE}\\
%\IEEEauthorblockA{Instituto de Telecomunicações\\
%Faculdade de Engenharia da Universidade do Porto, Portugal\\
%Email: samah.ghanem@fe.up.pt}
%\thanks{M. Shell is with the Department
%of Electrical and Computer Engineering, Georgia Institute of Technology, Atlanta,
%\thanks{Manuscript received April 19, 2005; revised January 11, 2007.}}
}
\maketitle

\begin{abstract}
%\boldmath
We investigate the optimal power allocation and optimal precoding for a cluster of two BSs which cooperate to jointly maximize the achievable rate for two users connecting to each BS in a MCP framework. This framework is modeled by a
virtual network MIMO channel due to the framework of full cooperation. In particular, due to sharing the CSI and data between the two BSs over the backhaul link. We provide a generalized fixed point equation of the optimal precoder in the asymptotic regimes of the low- and high-snr. We introduce a new iterative approach that leads to a closed-form expression for the optimal precoding matrix in the high-snr regime which is known to be an NP-hard problem. Two MCP distributed algorithms have been introduced, a power allocation algorithm for the UL, and a precoding algorithm
for the DL.
\end{abstract}


\IEEEpeerreviewmaketitle

\section{Introduction}

\IEEEPARstart{M}{ulti-cell} cooperative processing is well acknowledged for significantly improving spectral efficiency and fairness amongst users. Exploiting the concepts of cooperation via transmit diversity and virtualizing the networks can move networks into upper bounds despite the fundamental limits of cooperation \cite{96}. Therefore, cooperation in small size clustered network frameworks from one side can boost network performance. However, from another side, demonstrate that adaptivity and feedback can have a dramatic effect on the data rates when transmitter adapt to the channel experience. In this paper we consider a cooperative framework using mutli-cell processing~\cite{94}\footnote{A short conference version of this work is published in \cite{samahVTC13}.}. Firstly, we consider the problem of joint cooperative optimal power allocation. Therefore, with a prior knowledge of each channel state and the data from each UT, the base stations will cooperate to jointly design the optimal power allocation that maximizes the joint reliable information rates; i.e., the clustered network MIMO capacity. Each BS will then communicate the optimal power via feedback DL links to each user in order to use in their UL transmissions considering that the process is adaptive, and the processing time is very small, such that the CSI doesn't change.  The UL/DL reciprocity/duality is not assumed in solving the problem and there is no special setup like TDD considered, \cite{44}, \cite{95}. Secondly, we consider the problem of joint cooperative precoding for the DL scenario through which both BSs can jointly design the optimal precoding vectors. Besides, we provide insights into the way the interference can be mitigated from a precoding perspective and on the other hand, how a studied interference can be thought of as a positive factor in the MCP framework, instead of dealing with it as a limiting factor to the network capacity.

\vspace{0.5cm}

In this paper, we investigate the optimal power allocation and optimal precoding for a cluster of two BSs which cooperate to jointly maximize the achievable rate for two users connecting to each BS in a MCP framework. This framework is modeled by a virtual network MIMO channel due to the framework of full cooperation. In particular, due to sharing the CSI and data between the two BSs over the backhaul link. We provide a generalized fixed point equation of the optimal precoder in the asymptotic regimes of the low- and high-snr. We introduce a new iterative approach that leads to a closed-form expression for the optimal precoding matrix in the high-snr regime which is known to be an NP-hard problem. Two MCP distributed algorithms have been introduced, a power allocation algorithm for the UL, and a precoding algorithm for the DL. 

\section{The MCP System Model}

Consider the scenario shown in Figure~\ref{fig:figure-1} where MCP is implemented in clusters of two base stations. The base stations no longer tune their physical and link/MAC layer parameters separately (power level, time slots, sub-carrier usage, precoding coefficients etc.), but instead coordinate their coding and decoding operations on the basis of channel state information and user data information exchanged over a backhaul link~\cite{94}. \cite{DBLP:journals/corr/Ghanem14a} and \cite{samahWiMob16} considered a similar MCP scenario but with minimal cooperation, where data is not shared among BS, but only the CSI. The scenario considered here, exploits the MIMO MCP system considering no limits on the backhaul to which data and CSI can be shared among cooperative BSs. 

As illustrated in Figure~\ref{fig:figure-1}, we suppose that the base stations will share their CSI and will exchange the information received by the user terminals, UT1 and UT2 who are roaming under the coverage of BS1 and BS2, respectively. Therefore, BS1 and BS2 will receive from UT1 and UT2 respectively, 
\begin{figure}[ht!]
    \begin{center}
        \mbox{\includegraphics[height=3in,width=3in]{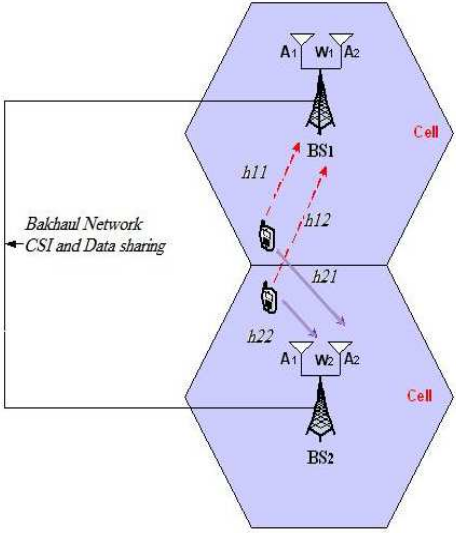}}
    \caption{MCP in a cluster of two base stations. The backhall link with finite bandwidth for sharing CSI and data is illustrated.}
    \label{fig:figure-1}
    \end{center}
    \label{figure-1}
\end{figure}
\begin{equation}
\label{a1}
{{y_{1}}}= \sqrt{snr}{h_{11} \sqrt{P_1} {{x_{1}}}}+ \sqrt{snr}h_{21} \sqrt{P_2} {{x_{2}}}+ {{n_{1}}}
\end{equation}
\begin{equation}
\label{a2}
{y_{2}}= \sqrt{snr} {h_{12} \sqrt{P_1} {{x_{1}}}}+\sqrt{snr} {h_{22} \sqrt{P_2} {{x_{1}}}+{{n_{2}}}}
\end{equation}
\normalsize 

${y_{1}}\in \mathbb{C}^n$ and ${y_{2}}\in \mathbb{C}^n$ represent the received vectors of complex symbols at BS1 and BS2 respectively, ${x_{1}}\in \mathbb{C}^n$ and ${x_{2}}\in \mathbb{C}^n$ represent the vectors of complex transmit symbols with zero mean and identity covariance ${\mathbb{E}[x_{1}x_{1}^{\dag}]},$ ${\mathbb{E}[x_{2}x_{2}^{\dag}]}$, respectively, ${n_{1}}\in \mathbb{C}^n$ and ${n_{2}}\in \mathbb{C}^n$ represent vectors of circularly symmetric complex Gaussian random variables with zero mean and identity covariance. ${h_{ij}}$ represent the complex gains of the sub-channels between transmitter ${i}$ and receiver ${j}$, where the main links are the ones with ${i=j}$, and the interference links are the ones with ${i} \neq {j}$. ${\sqrt{P1}}$ and ${\sqrt{P2}}$ represent the amplitude of the transmitted signals from UT1 and UT2, respectively. And $snr$ is the received signal to noise power ratio. The cooperation between the two base stations is incorporated via using the upper bound of the achievable rates in MIMO channels~\cite{35}, ICs~\cite{99}, and BCs~\cite{103}, as well as the MAC~\cite{99}. The achievable rates are:
\begin{equation}
{R_{1}} \leq I({{x_{1};y_{1}|x_{2}}})
\end{equation}
\begin{equation}
{R_{2}} \leq I({{x_{2};y_{2}|x_{1}}})
\end{equation}
\begin{equation}
{R_{1}+R_{2}}\leq \min~[I({{x_{1},x_{2};y_{1}}}),I({{x_{1},x_{2};y_{2}}})]
\end{equation}
\begin{equation}
\leq I({{x_{1},x_{2};y_{1},y_{2}}})
\end{equation}

Therefore, the optimization will be performed over the joint mutual information subject to the users power constraints, as follows:
\begin{equation}
\label{eq.a8.5}
\max \ I({{x_{1},x_{2};y_{1},y_{2}}})
\end{equation}
%\normalsize        
                       
Subject to:
\begin{equation}
\label{eq.a9.5}   
{P_1}\leq {Q_1}, \ {P_2}\leq {Q_2}, \ {P_1}\geq 0 \ and \ {P_2}\geq 0 
\end{equation}
%\normalsize  
    
Where ${P_1}$ and ${P_2}$ are the transmitted power corresponding to each UT, ${Q_1}$ and ${Q_2}$ is the total and maximum power each UT can use, the channels considered are scalar channels, and precoding is precluded in this UL scenario.

\section{Optimal Power Allocation with MCP}
MCP creates a virtual MIMO, that has its main power as a diagonal matrix, and it channels are the ones for the main channels gains and interference channels gains. Thus, we can represent the power for the two UTs as a diagonal matrix ${\bf{P}}=diag(P_1,P_2)$ and the channels can be represented as a $2 \times 2$ matrix with each row $i$ entries  
$h_{ii}, h_{ij}$. Thus, $vec({\bf{H}})=[h_{11}, h_{12}, h_{21}, h_{22}]$.

\subsubsection{Gaussian Inputs}
For Gaussian inputs, the mutual information is defined as:
\begin{equation}
\label{eq.a11.5}
{I({{x_{1},x_{2};y_{1},y_{2}}})}={log \left|\left(\bf{H}\bf{P}\bf{P}^{\dag}\bf{H}^{\dag}+\bf{I}\right)\right|}
\end{equation}

\vspace{0.2cm}

\begin{theorem}
\label{theorem1.5}
The optimal power allocation for two UTs in the MCP framework ${(P_1^\star,P_2^\star)}$ with Gaussian inputs in terms of channel coefficients and powers follows the following form:
\begin{equation}
\label{eq.a12.5}
\left\{\begin{array}{l l}
{P_1^{\star}=Q_1},\\ 
{P_2^{\star}=Q_2},
\end{array} \right.
\end{equation}
\end {theorem}

\begin{proof}
Theorem~\ref{theorem1.5} follows the solution of the KKT conditions of~\eqref{eq.a8.5} subject to~\eqref{eq.a9.5}, and due to the fact that the function is increasing with respect to the power -the matrix within the log is a positive definite matrix. Thus, we notice that the solution of the derivative with respect to ${P_1}$ leads to ${P_2^{\star}}$, and the derivative of ${P_2}$ leads to ${P_1^{\star}}$.
\end{proof}

\vspace{0.3cm}

It can be easily verified that~\eqref{eq.a11.5} is concave with respect to each user main power since the second derivative is always negative, and also through the positive definiteness of the matrix in~\eqref{eq.a11.5}. Capitalizing on the relation between the gradient of the mutual information and the MMSE in~\cite{35}, we re-investigate the result in the context of the MCP cooperative framework. The relation between the gradient of the mutual information in~\eqref{eq.a8.5} and the MMSE is as follows:
\begin{equation}
\label{mimo}
\nabla_{P} I({x_{1},x_{2};y_{1},y_{2}})=\bf{H}^{\dag}\bf{H}\bf{P}\bf{E}
\end{equation}
\begin{equation}
=\left[\begin{matrix} \nabla_{P_{11}}I({x_{1},x_{2};y_{1},y_{2}}) & \nabla_{P_{12}}I({x_{1},x_{2};y_{1},y_{2}}) \\ \nabla_{P_{21}}I({x_{1},x_{2};y_{1},y_{2}}) & \nabla_{P_{22}}I({x_{1},x_{2};y_{1},y_{2}}) \end{matrix}\right],
\end{equation}

where:
\begin{multline}
\nabla_{P11}I(x_{1},x_{2};y_{1},y_{2})={h_{11}}^{*}h_{11} \sqrt{P_1}E_{11}+{h_{21}}^{*}h_{21} \sqrt{P_1}E_{11}\\
+{h_{11}}^{*}h_{12} \sqrt{P_2}E_{21}+{h_{21}}^{*}h_{22} \sqrt{P_2}E_{21}   
\end{multline} 
\begin{multline}
\nabla_{P12}I(x_{1},x_{2};y_{1},y_{2})={h_{11}}^{*}h_{11} \sqrt{P_1}E_{12}+{h_{21}}^{*}h_{21} \sqrt{P_1}E_{12}\\
+{h_{11}}^{*}h_{12} \sqrt{P_2}E_{22}+{h_{21}}^{*}h_{22} \sqrt{P_2}E_{22}   
\end{multline}
\begin{multline}
\nabla_{P21}I(x_{1},x_{2};y_{1},y_{2})={h_{12}}^{*}h_{11} \sqrt{P_1}E_{11}+{h_{22}}^{*}h_{21} \sqrt{P_1}E_{11}\\
+{h_{12}}^{*}h_{12} \sqrt{P_2}E_{21}+{h_{22}}^{*}h_{22} \sqrt{P_2}E_{21}   
\end{multline}
% 
%\vspace{-0.6cm} 
%%\renewcommand{\theequation}{5.16}
\begin{multline}
\nabla_{P22}I(x_{1},x_{2};y_{1},y_{2})={h_{12}}^{*}h_{11} \sqrt{P_1}E_{12}+{h_{22}}^{*}h_{21} \sqrt{P_1}E_{12}\\
+{h_{12}}^{*}h_{12} \sqrt{P_2}E_{22}+{h_{22}}^{*}h_{22} \sqrt{P_2}E_{22} .  
\end{multline}

The MMSE matrix $\bf{E}$ defines the elements of the gradient of the mutual information with respect to the main links and interference links powers, as follows:
\begin{equation}
\bf{E}=\bf{\left[\begin{matrix} E_{11} & E_{12} \\ E_{21} & E_{22} \end{matrix}\right]},
\end{equation}

with the expansion of $\bf{E}$ is given by:
\begin{equation}
E_{11}=\mathbb{E}[(x_{1}-\mathbb{E}(x_{1}|y_{1},y_{2}))(x_{1}-\mathbb{E}(x_{1}|y_{1},y_{2}))^{\dag}]
\end{equation}
%\vspace{-0.8cm} 
%%\renewcommand{\theequation}{5.19}
\begin{equation}
E_{12}=\mathbb{E}[(x_{1}-\mathbb{E}(x_{1}|y_{1},y_{2}))(x_{2}-\mathbb{E}(x_{2}|y_{1},y_{2}))^{\dag}]
\end{equation}
%\vspace{-0.6cm} 
%%\renewcommand{\theequation}{5.20}
\begin{equation}
E_{21}=\mathbb{E}[(x_{2}-\mathbb{E}(x_{2}|y_{1},y_{2}))(x_{1}-\mathbb{E}(x_{1}|y_{1},y_{2}))^{\dag}]
\end{equation}
%\vspace{-0.6cm} 
%%\renewcommand{\theequation}{5.21}
\begin{equation}
E_{22}=\mathbb{E}[(x_{2}-\mathbb{E}(x_{2}|y_{1},y_{2}))(x_{2}-\mathbb{E}(x_{2}|y_{1},y_{2}))^{\dag}]
\end{equation}

$E_{11}$ and $E_{22}$ correspond to MMSE1 and MMSE2, respectively; that is the per-user MMSE which defines the error in each main link, and their sum is the total error. However, $E_{12}$ and $E_{21}$ are covariance functions of the estimates of the decoded symbols for each UT. Note that the non-linear estimates of each user input is given as follows:
\begin{multline}
\widehat{x}_1={{\mathbb{E}[x_{1}|y_1,y_2]}}=\\
\sum_{{x_{1},x_{2}}} \frac{{{x_{1}}}p_{y_1,y_2|x_{1},x_{2}}({{y_1,y_2|x_{1},x_{2}}})p_{x_1}({{x_{1}}})p_{x_2}({{x_{2}}})}{p_{y}({{y_1,y_2}})}
\end{multline}
\begin{multline}
\widehat{x}_2={\mathbb{E}[x_{2}|y_1,y_2]}=\\
\sum_{{{x_{1},x_{2}}}} \frac{{{x_{2}}}p_{y_1,y_2|x_{1},x_{2}}({{y_1,y_2|x_{1},x_{2}}})p_{x_1}({{x_{1}})p_{x_2}({x_{2}}})}{p_{y_1,y_2}({{y_1,y_2}})}.
\end{multline}

The non-linear estimates give a statistical intuition to the problem, however, in practical setups, such estimates can be found via the linear MMSE. In particular, the inputs estimates can be found by deriving the optimal Wiener receive filters solving a minimization optimization problem of the MMSE, the following theorem provides the linear estimates.

\begin{theorem}
The linear estimates of the inputs of the vectors ${\bf{x}}=[{x_1},{x_2}]$ for each user in the virtual MIMO MCP framework given the output vector ${\bf{y}}=[{{y_1},{y_2}}]$ can be expressed as:
\begin{equation}
{\bf{\widehat{x}}}=\bf{P^{\dag}H^{\dag}(I+P^{\dag}H^{\dag}PH)^{-1}y}
\end{equation}
\end{theorem}

\begin{proof}
The proof of this theorem follows from the derivative of the linear MMSE to derive the Wiener filter part as the MMSE minimizer.
\end{proof}

\vspace{0.3cm}

\subsubsection{Arbitrary Inputs}
There are no closed-form expressions for the mutual information with arbitrary inputs; therefore, we need to capitalize on the relation between the gradient of the mutual information and the MMSE to derive the optimal power allocation for the generalized inputs.

\begin{theorem}
\label{theorem2.5}
The optimal power allocation for two UTs in the MCP framework ${(P_1^\star,P_2^\star)}$ with arbitrary inputs -in terms of channel coefficients, and the MMSE- takes the following form:
\begin{multline}
\label{13.5}
{\lambda_1}^{\star}\sqrt{P_1}=({h_{11}}^{*}h_{11}+{h_{21}}^{*}h_{21})\sqrt{P_1}E_{11} \\
+({h_{11}}^{*}h_{12}+{h_{21}}^{*}h_{22})\sqrt{P_2}E_{21}  
\end{multline}
\begin{multline}
\label{14.5}
{\lambda_2}^{\star}\sqrt{P_2}=({h_{12}}^{*}h_{11}+{h_{22}}^{*}h_{21})\sqrt{P_1}E_{12} \\
+({h_{12}}^{*}h_{12}+{h_{22}}^{*}h_{22})\sqrt{P_2}E_{22}\end{multline}
\end{theorem}

\begin{proof}
See Appendix A.
\end{proof}

\vspace{0.3cm}

Theorem~\ref{theorem2.5} can be solved numerically to search the optimal power allocation of both users, where ${\lambda_1}$ and ${\lambda_2}$ are the Lagrange multipliers, it assimilates a mercury/waterfilling for the arbitrary inputs that compensate for the non-Gaussianess of the binary constellations, and a waterfilling for the Gaussian inputs, where more power is allotted to less noisy channels. Moreover, when both user powers are non-zero, we can re-write Theorem~\ref{theorem2.5} with respect to the MMSE and the covariance as follows:
%%\renewcommand{\theequation}{5.27}
%\small
\begin{multline}
\label{15.5}
P_1^{\star}= \\
\frac{1}{snr1|h_{11}^{*}h_{11}+ h_{21}^{*}h_{21}|}~mmse~(snr1|h_{11}^{*}h_{11}+h_{21}^{*}h_{21}|P_1^{\star})\\+
\frac{1}{snr2|h_{11}^{*}h_{12}+ h_{21}^{*}h_{22}|}~cov~(snr2|h_{11}^{*}h_{12}+ h_{21}^{*}h_{22}|P_2^{\star}),
\end{multline}
%\normalsize
and,
\begin{multline}
\label{16.5}
P_2^{\star}= \\
\frac{1}{snr2|h_{12}^{*}h_{12}+h_{22}^{*}h_{22}|}~mmse~(snr2|h_{12}^{*}h_{12}+ h_{22}^{*}h_{22}|P_2^{\star})\\+
\frac{1}{snr1|h_{12}^{*}h_{11}+h_{22}^{*}h_{21}|}~cov~(snr1|h_{12}^{*}h_{11}+h_{22}^{*}h_{21}|P_1^{\star}).
\end{multline}
%\normalsize

Its straightforward to see that for the case when the inputs are time-division multiplexed, the optimal power allocation takes the form: ${P_2^\star=Q_2}$ when ${P_1=0}$, and ${P_1^\star=Q_1}$ when ${P_2=0}$. In addition, we can easily specialize the result of \eqref{13.5} and \eqref{14.5} to the one in \eqref{eq.a12.5} for Gaussian inputs. In particular, we substitute the linear MMSE for Gaussian inputs given by:
\begin{equation}
{\bf{E}}=({\bf{P}^{\dag}\bf{H}^{\dag}\bf{H}\bf{P}}+{\bf{I}})^{-1}
\end{equation}

into \eqref{13.5} and \eqref{14.5}, it can be easily shown that the optimal power allocation in Theorem~\ref{theorem2.5} matches the one in Theorem~\ref{theorem1.5}. However, it is worth to notice other solutions that can be derived from \eqref{15.5} and \eqref{16.5} for other setups, like the two-user MAC-channel, see \cite{98} and \cite{DBLP:journals/corr/Ghanem14}. 

\section{Optimal Precoding with MCP}
We consider the MCP cooperation in the DL where both BSs jointly cooperate to design the optimal precoding vectors that maximize their achievable rates. The optimization problem stays the same, this choice is convenient due to the fact that the joint mutual information upper bounds the broadcast framework, see~\cite{103}. The following theorem gives a  generalization of the optimal precoder structure at the low- and high-snr regime. In particular, we will show that the precoder should admit a structure that performs matching of the strongest source modes to the weakest noise modes, and this alignment enforces a permutation process to appear in the power allocation.     

\begin{theorem}
\label{theorem3.5}
The non-unique first-order optimal precoder that maximizes the mutual information with the MCP that substitutes a MIMO setup subject to an average power constraint can be written as follows:
\begin{equation}
\label{a18} 
\bf{P}^{\star}=\bf{U}\bf{D}\bf{R}^{\dag}  
\end{equation}
\end{theorem}
Where $\bf{U}$ is a unitary matrix, $\bf{D}$ is a diagonal matrix, and $\bf{R}$ is a rotation matrix.
\begin{proof}
Theorem~\ref{theorem3.5} follows the relation between the gradient of the mutual information and the MMSE and the decomposition of its matrix components, see Appendix B.
\end{proof}

\vspace{0.3cm}

\section{The Asymptotic Regimes}
The following sections will specialize the study of the MCP framework to key asymptotic regimes of the SNR, particularly, the low-snr and the high-snr. One of the interesting observations that follows an in depth analysis of both regimes is that the optimal designs for both the low- and high-snr performs a diagonalization operation to at least one of the system elements that is causing correlation among different system variables; whether the correlation is among the sub-channels and so diagonalizing the channel matrix or among the inputs and so diagonalizing the error matrix. However, the optimal precoder is not necessarily diagonal. In fact, for Gaussian inputs, the optimal precoder is a diagonal matrix. However, the optimal precoder for arbitrary inputs is a non-diagonal matrix.

\subsubsection{The Low-SNR Regime}
Consider the analysis of the optimal power allocation and optimal precoding with MCP capitalizing on the low-snr expansions of the conditional probability distribution of the Gaussian noise defined as:
\begin{equation}
\label{a19}
p_{y_1,y_2|x_1,x_2}({{y_1,y_2|x_1,x_2}})=\frac{1}{\pi^{n}} e^{-\left\|\left[\begin{matrix} y_1 \\ y_2 \end{matrix}\right]-\bf{H}\bf{P}\left[\begin{matrix} x_1 \\ x_2 \end{matrix}\right]\right\|^{2}}
\end{equation}
%\normalsize

And the MMSE defined as:
\begin{equation}
\label{eq.a20.5}
{MMSE(snr)=\mathbb{E}\left[\left\|\bf{HP}\left(\left[\begin{matrix} x_1 \\ x_2 \end{matrix}\right]-\mathbb{E}\left[\begin{matrix} x_1|y_1,y_2 \\ x_2|y_1,y_2 \end{matrix}\right]\right)\right\|^{2}\right]}
\end{equation}
%\normalsize

The low-snr expansion of the MMSE matrix can be expressed as follows:
\begin{equation}
\label{a21}
{\bf{E}}={\bf{I}}-({\bf{HP}})^{\dag}{\bf{HP}}.snr+\mathcal{O}(snr^{2}) 
\end{equation}
%\normalsize

Consequently, the low-snr expansion of the non-linear MMSE is given by the following theorem.
\begin{theorem}
\label{theorem4.5}
The low-snr expansion of the non-linear MMSE in ~\eqref{eq.a20.5} as $snr \to 0$ is given by:
\begin{multline}
\label{a22}
MMSE(snr)=Tr\left\{\bf{HP}{(\bf{HP})}^{\dag}\right\}\\
-Tr\left\{(\bf{HP}({\bf{HP}}^{\dag}))^{2}\right\}.snr+\mathcal{O}(snr^{2})     
\end{multline}
%\normalsize
\end{theorem}

\begin{proof}
See Appendix C.
\end{proof}

\vspace{0.3cm}

Now, by virtue of the relationship between mutual information and MMSE, the Taylor low-snr expansion of the mutual information is given by:
%%\renewcommand{\theequation}{5.35}
%\footnotesize 
\begin{multline}
\label{a23}
I(snr)=Tr \left\{\bf{HP}(\bf{HP})^{\dag}\right\}.snr \\
-Tr\left\{({\bf{HP}(\bf{HP}^{\dag}}))^{2}\right\}.\frac{snr^{2}}{2}+\mathcal{O}(snr^{3})      
\end{multline}
%\normalsize

According to~\eqref{a23}, for first-order optimality, the form of the optimal precoder -in the DL- follows from the low-snr expansions the form of optimal precoder for the complex Gaussian inputs settings. To prove this claim, we will re-define our optimization problem as follows:
\begin{equation}
\label{eq.a24.5}
max \ Tr\left\{\bf{HP}(\bf{HP})^{\dag}\right\}.snr      
\end{equation}

Subject to:
\begin{equation}
\label{eq.a25.5}
Tr\left\{\bf{PP}^{\dag}\right\}\leq1   
\end{equation}

Then, we do an eigen value decomposition such that: $\bf{H}^{\dag}\bf{H}=\bf{U}\Omega\bf{U}^{\dag}$. Let $\tilde{\bf{P}}=\bf{U}^{\dag}\bf{P}$, and let $\bf{Z}\succeq 0$ is a positive semi-definite matrix, such that: $\bf{Z}=\tilde{\bf{P}}{\tilde{\bf{P}}}^{\dag}$. Then, the optimization problem can be re-written as: $max~~Tr\left\{\bf{Z}\Omega \right\}$ subject to: $Tr\left\{\bf{Z}\right\}\leq1$, which leads to the solution, $\bf{Z}=\lambda^{-1}\Omega snr$. In the DL, this result proves that the optimal precoder in the low-snr performs mainly two operations, firstly, it aligns the transmit directions with the eigenvectors of each user sub-channel. Secondly, it performs power allocation over the user sub-channels; i.e., the main and the interference links. Moreover, in the UL, it can be easily shown that specializing the low-snr results to the Gaussian inputs case by deriving the Taylor expansion of~\eqref{eq.a11.5} as $snr \to 0$ will follow the one in \eqref{a23} for the general inputs. It follows that the optimal power allocation as $snr \to 0$, for any inputs regardless of their signaling will follow the one for the Gaussian inputs in~\eqref{eq.a12.5}. Consequently, the mutual information is insensitive to the distribution of the inputs signaling in the low-snr. To prove this claim, we substitute all channel states and powers into \eqref{eq.a24.5}, and the optimization problem will be as follows:
%\vspace{-0.2cm}
%%\renewcommand{\theequation}{5.38}
\begin{equation}
\label{eq.a32.5.new}
max \left\{{h_{11}}^{2}P_1+{h_{12}}^{2}P_2+{h_{21}}^{2}P_1+{h_{22}}^{2}P_2\right\}. snr	                      
\end{equation}
Subject to:  
\begin{equation}
P_1\leq Q_1	                      
\end{equation}
%\vspace{-0.2cm}
%%\renewcommand{\theequation}{5.40}
\begin{equation}
P_2\leq Q_2	                      
\end{equation}
It follows that the gradient of the mutual information in~\eqref{eq.a32.5.new} with respect to the input powers is only a function of the channel states and the snr, that is:
\begin{equation}
\label{eq.a35.5.new}
\nabla_{P_1}I({{x_{1},x_{2};y_{1},y_{2}}})=\left\{h_{11}^{2}+h_{12}^{2}\right\}snr=\lambda_1\sqrt{P_1}	            
\end{equation}
and,
%\vspace{-0.2cm}
%%\renewcommand{\theequation}{5.42}
\begin{equation}
\label{eq.a36.5.new}
\nabla_{P_2}I({{x_{1},x_{2};y_{1},y_{2}}})=\left\{h_{21}^{2}+h_{22}^{2}\right\}snr=\lambda_2\sqrt{P_2}                    
\end{equation}
Therefore, when $snr \to 0$, we can write the result~\eqref{eq.a35.5.new} and~\eqref{eq.a36.5.new} in a matrix formulation, see [Eq.48,~\cite{100}], as follows:
\begin{equation}
D_{{\textbf{P}}}I({{x_{1},x_{2};y_{1},y_{2}}})=vec \left({\bf{H}^{\dag}\bf{H}}\right)	                      
\end{equation}
And this proves our claim.

\subsubsection{The High-SNR Regime}
The characterization of the optimal precoder at the high-snr is known to be an NP-hard problem. In fact, the lack of explicit expressions for the capacity of binary input constellations makes the goal more difficult. In \cite{34}, they provide an explicit expression of the capacity of SISO Gaussian channels with BPSK inputs and the MMSE counter part. They verified the fundamental relation between the mutual information and the MMSE for the special case of Gaussian inputs to the BPSK signaling, therefore, its proved for general inputs. The MMSE and the mutual information for BPSK signaling over SISO AWGN channel, respectively, are obtained as:
\begin{equation}
\label{m1}
{mmse(snr)}=1-\int_{-\infty}^{\infty}\frac{e^{-({\zeta-\sqrt{snr}})^2}}{\sqrt{\pi}}tanh(2\sqrt{{snr}}{\zeta}) d{\zeta}
\end{equation}
\begin{equation}
\label{m2}
{I(snr)}={snr}-\int_{-\infty}^{\infty}\frac{e^{-({\zeta-\sqrt{snr}})^2}}{\sqrt{\pi}}log~cosh(2\sqrt{{snr}}{\zeta}) d{\zeta}
\end{equation}
The authors in~\cite{34} didn't provide a detailed proof of their result, therefore, we present the derivation in Appendix D. The $tanh$ term corresponds to the conditional mean estimate of the input or the non-linear estimate $\mathbb{E}[\bf{x|y}]$.  Moreover, we can also see how the mmse in~\eqref{m1} relates to the error function, or in other words, to the probability of bit error rate of BPSK inputs under AWGN channel as follows:
\begin{equation}
{mmse(snr)}=1-\int_{-\infty}^{\infty}\frac{e^{-({\zeta-\sqrt{snr}})^2}}{\sqrt{\pi}}tanh(2\sqrt{{snr}}{\zeta}) d{\zeta}
\end{equation}
%
%\vspace{-0.5cm}
%%\renewcommand{\theequation}{5.47}
\begin{equation}
\geq 1-\int_{0}^{\infty}\frac{e^{-({\zeta-\sqrt{snr}})^2}}{\sqrt{\pi}}d{\zeta}
\end{equation}
%\vspace{-0.5cm}
%%\renewcommand{\theequation}{5.48}
\begin{equation}
=1-\int_{\sqrt{snr}}^{\infty}\frac{e^{-({\zeta})^2}}{\sqrt{\pi}}d{\zeta}=\frac{1}{2}erfc\left(\sqrt{snr}\right)
\end{equation}
In addition, its worth to observe the geometric properties of the solution in~\eqref{m2}. The cosh hyperbolic function defines the decision regions of the constellation points detection multiplied by a two sided error function, through which the mutual information reaches a saturation limit of $log_{2}(2)=1$; a normalized snr when the second term goes to zero, which makes any binary constellation matches the Gaussian one in terms of mutual information at the low-snr, see \cite{SamahMaptele}. Nonetheless, \eqref{m1} and \eqref{m2} can be multiplied by 2 to gain the closed-form expressions of QPSK inputs since the decision regions of the hyperbolic function extends over the real and imaginary parts with constellation points ${\bf{x}}_{QPSK}=\{1+1j,1-1j,-1-1j,-1+1j\}$ instead of ${\bf{x}}_{BPSK}=\{1,-1\}$. This geometric interpretation of the solution may help advancing future research to find closed form expressions for the mutual information of other types of binary constellations as well as multi-user setups. Moreover, the mutual information for SISO Gaussian channels with BPSK input distribution is expanded for high-snr and upper- and lower-bounded in terms of the minimum transmit lattice distance ${d_{min}}$ and maximum receive lattice distance ${d_{max}}$ between the constellation points, see~[Theorem 4,~\cite{3}]. Therefore, capitalizing on the result in \cite{3}, we can derive the structure of the optimal precoder for each user terminal in the MCP setup. We can define the optimization problem using the upper bound as follows:
\begin{equation}
\label{eq.a37.5}
\max \ {log~M-\frac{e^{-{d_{min}}^{2} \frac{snr}{4}}}{Md_{min} snr}\big(\sqrt{\pi}-\frac{4.37+2\sqrt{\pi}}{{d_{min}}^{2} snr}\big)}
\end{equation}

Subject to:
\begin{equation}
\label{a38}
Tr \left\{\bf{P}\bf{P}^{\dag}\right\}=1	                      
\end{equation}
%\normalsize

With ${M}$ is the product of the constellation cardinality and ${d_{min}}$ is the the minimum distance between the ${M}$ possible realizations of the input vector of the constellation, therefore, its defined as:
\begin{equation} 
\label{eq.aa38.5}
{d_{min}}=min_{i\neq j}\bf{\left\|HP\left(x_i-x_j\right)\right\|}
\end{equation}

However, due to \eqref{eq.aa38.5} and due to the fact that there is no explicit form for the optimal precoder in the high-snr regime, and by virtue of [Eq.16,~\cite{97}], which has been identified as an NP-hard problem, we define the initial value in the numerical solution of ${d_{min}}$ as follows:
\begin{equation}
\label{eq.a39.5}
{d_{min}}= \bf{{(HP)}}^{\dag}\bf{H}\bf{P}                      
\end{equation}

\begin{theorem}
\label{theorem6.5}
The optimal precoder matrix in the high-snr for a BPSK constellation in the MCP setup is the solution of:
\begin{multline}
\label{eq.a41.5}
D_{\bf{P}} I(x_{1},x_{2};y_{1},y_{2})=\\
-a\big((\bf{G}^{-1}\bf{VW})^{T}\otimes(\bf{G}^{-1} \bf{P}^{\dag} \bf{C})\big)
+a\big((\bf{G}\bf{W})^{T}\otimes(\bf{U}\bf{V}\bf{P}^{\dag}C)\big)\\
+a\big(\bf{W}^{T}(\bf{U}\bf{V}\bf{G}\bf{P}^{\dag}\bf{C})\big)
+ab\big((\bf{G}^{-2})^{T}\otimes(\bf{U}\bf{V}\bf{G}^{-1}P^{\dag}\bf{C})\big)\\
+ab\big((\bf{G}^{-1})^{T}\otimes(\bf{U}\bf{V}\bf{G}^{-1} P^{\dag}\bf{C})\big)=0                   
\end{multline}

with $a=\frac{\sqrt{pi}}{2M}$ , $b=\frac{4.37+2\sqrt{pi}}{4\sqrt{pi}}$, $\bf{C}=\frac{snr}{4}\bf{H}^{\dag}\bf{H}$, $\bf{G}=\bf{P}^{\dag}\bf{C}\bf{P}$, $\bf{U}=\bf{G}^{-1}$, $\bf{V}=e^{-G^{2}}$, and ${\bf{W}}={\bf{I}}-b\bf{G}^{-2}$.
\end{theorem}

\begin{proof}
Substitute~\eqref{eq.a39.5} into ~\eqref{eq.a37.5} and capitalizing on the matrix differentiation theories in~\cite{101}, the theorem can be proved\footnote{Details of proof is removed to reduce redundancy.}.
\end{proof}

\vspace{0.3cm}

Solving~\eqref{eq.a41.5} numerically, we can see that the optimal precoding matrix is always a non-diagonal matrix.

\section{MCP Distributed Algorithms}
Comparing the cooperative framework using MCP - which models a MIMO channel - to the non-cooperative framework - which models an interference channel - we can analytically understand the benefits and drawbacks of each framework. In particular, the achievable rates via the cooperation is higher than that without cooperation. However, the processing overload and the CSI and data exchange overhead is another tradeoff. If the interference is orthogonal to the main channel, then we can preclude the interference effect, which can be dealt with through receive antenna diversity, and therefore we can maximize the information rates in the UL. However, we can attack the interference problem in the DL via adding a studied interference, i.e., via aligning the interference, or via precoding; as proposed in Theorem~\ref{theorem3.5}. We will introduce the MCP distributed algorithms, the first algorithm gives the optimal power allocation for the UL, and the second algorithm gives the optimal precoding for the DL.
\tiny
\begin{algorithm}[h!]
		\SetKwInOut{Input}{Input}\SetKwInOut{Output}{Output}
			\DontPrintSemicolon
            \BlankLine			
           %\BlankLine
           %\vspace{-0.35cm}
%\normalsize           
          \small  
          $\bf{BS1}$ {\bf{Input:}} $CSI1$, $\mathbb{E}(x_1|y_1)$, $\mathbb{E}(x_2|y_1)$     
          \BlankLine        
          $\bf{BS2}$ {\bf{Input:}} $CSI2$, $\mathbb{E}(x_1|y_2)$, $\mathbb{E}(x_2|y_2)$
            \BlankLine			
            \BlankLine				
		\uIf{$\bf{BW}$  $\bf{Backhaul}$ $\geq Threshold~{\tau} $}{
		\BlankLine
		%\small
		$\bf{BS1}$ and $\bf{BS2}$ declare congestion and minimal cooperation message}
		\uElse{
			$\bf{BS1}$ sends decoded $x_1: \ \mathbb{E}(x_1|y_1)$ and $CSI1$ to $\bf{BS2}$
			\BlankLine
      $\bf{BS2}$ sends decoded $x_2: \ \mathbb{E}(x_2|y_2)$ and $CSI2$ to $\bf{BS1}$
      \BlankLine
      $\bf{BS1}$ and $\bf{BS2}$ check resources$\to$ handshaking$\to$ $BS1/BS2$ will do the processing.}      

      {\bf{Output:}}{~~The optimum power allocation in the UL is the solution for:
        \BlankLine
        ~~~~~~~~~~~~~~~~~~~~~~~$\bf{P}_{k+1}=\alpha_{k} \bf{P}_{k}+ \alpha_{k} \lambda \bf{H}^{\dag} \bf{H}\bf{P}_{k} E_{k}$
         \BlankLine
        For the two UT case,
         \vspace{-0.3cm}      
\begin{center}
${P_1}^{\star}={\bf{P}}(1,1)$
        \BlankLine
        ${P_2}^{\star}={\bf{P}}(2,2)$
\end{center}
        %\vspace{-.02cm}
       $\bf{BS1}$ and $\bf{BS2}$ share  ${P_1}^{\star}$ and ${P_2}^{\star}$ and feedback to $UT1$ and $UT2$. 
 }     						
\caption{\;
Optimum Power Allocation with MCP-Uplink\;
%%\vspace{.03cm}
Full cooperation: CSI and data sharing \cite{107}}
\label{algo_1}
\normalsize
\end{algorithm}
\normalsize
%\vspace{-0.7cm}
%%%%%%%%%%%%%%%%%%%%%%%%%%%%%%%%%%%%%%%%%%%%%%%%%%%%%%%%%%%% ALGORITHM-II
\tiny
\begin{algorithm}[h!]
			\SetKwInOut{Input}{Input}\SetKwInOut{Output}{Output}
			\DontPrintSemicolon
            \BlankLine			
           % \BlankLine
           %\vspace{-0.35cm}
\small      
         $\bf{BS1}$ {\bf{Input:}} $CSI1$, $x_1$, $x_2$
          \BlankLine
          $\bf{BS2}$ {\bf{Input:}} $CSI2$, $x_1$, $x_2$
            \BlankLine			
           % \BlankLine
			 		BS1 and BS2 perform SVD($\bf{H}$):~~$\bf{H}=\bf{U}_{H} \Lambda_{H} \bf{V}_{H}^{\dag}$
		 		%\BlankLine
        %  $\bf{H}=\bf{U}_{H} \Lambda_{H} \bf{V}_{H}^{T}$  
        % \BlankLine
\begin{center}
$BS1$ sends $(h_{21}\nu_{h_{11}}\sqrt{P_1}+ h_{22}\nu_{h_{21}}\sqrt{P_1})x_{1}$ to $BS2$
\BlankLine
$BS2$ sends $(h_{11}\nu_{h_{12}}\sqrt{P_2}+ h_{12}\nu_{h_{22}}\sqrt{P_2})x_{2}$ to $BS1$
\end{center}
   {\bf{Output:}}{~~The optimum precoding in the DL is done via each BS solving,
 \BlankLine        
\begin{center}
%$\bf{P}_{k}=\bf{V}_{H} \bf{P}_{k+1}$
${\bf{P}_{k}}={\bf{V}_{H}} diag(\sqrt{P_1},\sqrt{P_2})$
\BlankLine 
${\bf{P}_{k+1}}=\alpha_{k} {\bf{P}_{k}}+ \alpha_{k} \lambda \bf{H}^{\dag} \bf{H}\bf{P}_{k} E_{k}$
\end{center}
%For a two simultaneous transmissions from each base station,
%\BlankLine
$\bf{BS1~~transmits}:$
\BlankLine 
\begin{center}
$(h_{11}\nu_{h_{11}} \sqrt{P_1}+ h_{12}\nu_{h_{21}} \sqrt{P_1})x_{1}+(h_{11}\nu_{h_{12}} \sqrt{P_2}+h_{12}\nu_{h_{22}} \sqrt{P_2})x_{2}$  
\end{center}
%\BlankLine
$\bf{BS2~~transmits}:$
\BlankLine 
\begin{center}
$(h_{21}\nu_{h_{11}} \sqrt{P_1}+ h_{22}\nu_{h_{22}} \sqrt{P_1})x_{1}+(h_{21}\nu_{h_{12}} \sqrt{P_2}+h_{22}\nu_{h_{22}} \sqrt{P_2})x_{2}$
\end{center}
} 						
The process will be iteratively repeated for each simultaneous transmission of $BS1$ and $BS2$.
\caption{\;
Optimum Precoding with MCP-Downlink \;
%\vspace{.03cm}
Full cooperation: CSI and data sharing \cite{107}
}
\normalsize
\label{algo_2}
%\normalsize
\end{algorithm}
\normalsize

\section{Numerical Analysis}
We shall now introduce a set of illustrative results that cast insight into the problem. The results for the Gaussian inputs setup is straightforward with the mutual information closed form. However, we used Monte-Carlo method to generate the achievable rates for arbitrary inputs.  Its of particular relevance to notice that the Gaussian inputs distribution is optimal compared to the arbitrary inputs distribution from the rate achievability sense, as shown in Figure~\ref{fig:figure-4.5}, and Figure~\ref{fig:figure-5.5}. We can easily verify that for the same transmit power, higher achievable rates are possible with Gaussian inputs. Moreover, the arbitrary inputs may lie at a certain point at the null space of the channel causing a decay in the achievable rates. For the Gaussian inputs, the optimal power allocation chosen by each user is to use their own maximum power, as illustrated in Figure~\ref{fig:figure-6.5}, therefore, this serves to maximize the data rates in the UL and DL scenarios. However, for the case of arbitrary inputs, the optimality is to search for a set where both inputs don't lie in the null space of the channel - Voronoi region - therefore, they don't cancel each other. Hence, optimal power allocation is not a sufficient solution, therefore, we can improve the decay in the mutual information either by orthogonalizing the inputs, or via precoding them. In addition, Figure~\ref{fig:figure-8.5} illustrates the main ideas of interference and interference free channels with respect to the mutual information and the errors. Notice that the channel gains are chosen to be unity for main and interference links when interference is considered. This will assure that the channel will not amplify nor attenuate the transmitted signals. Therefore, we can see that the loss in the achievable rate is 0.5 bits, such that without interference the achievable rate is 2 bits, and with interference the achievable rate is 1.5 bits. However, the 0.5 bit loss is induced through $E_{12}$ and $E_{21}$ causing $E_{11}+E_{22}$ to saturate at 0.5 instead of zero.
\begin{figure}[ht!]
    \begin{center}
        \mbox{\includegraphics[height=2.2in,width=3in]{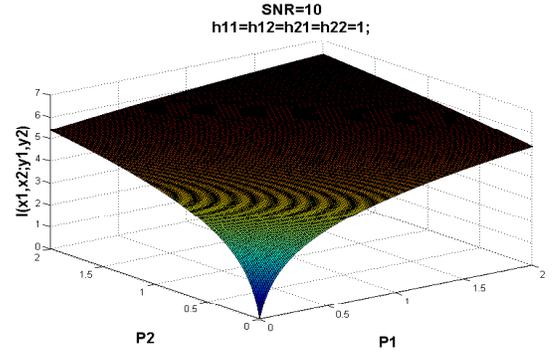}}     
        \caption{The achievable rate for the MCP with respect to UTs main power with Gaussian inputs.}
    \label{fig:figure-4.5}
    \end{center}
    \label{figure-4.5}
\end{figure}
\begin{figure}[ht!]
    \begin{center}
        \mbox{\includegraphics[height=2in,width=2.95in]{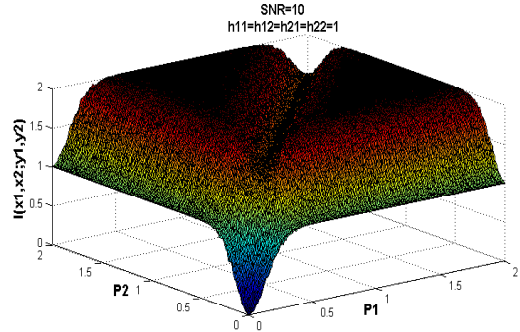}}               
    \caption{The achievable rate for the MCP with respect to UTs main power with BPSK inputs.}
    \label{fig:figure-5.5}
    \end{center}
    \label{figure-5.5}
\end{figure}
\begin{figure}[ht!]
    \begin{center}
        \mbox{\includegraphics[height=2in,width=2.8in]{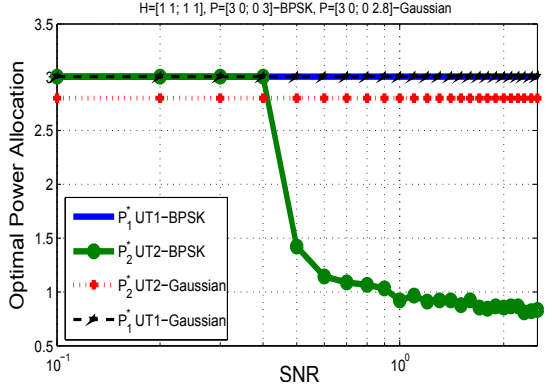}}
     \caption{Optimal power allocation for Gaussian inputs and BPSK inputs.}
    \label{fig:figure-6.5}
    \end{center}
\end{figure}
%%
%%\vspace{0.5cm}
%%
%%\renewcommand{\thefigure}{5.5}
\begin{figure}[ht!]
    \begin{center}
        \mbox{\includegraphics[height=2in,width=2.9in]{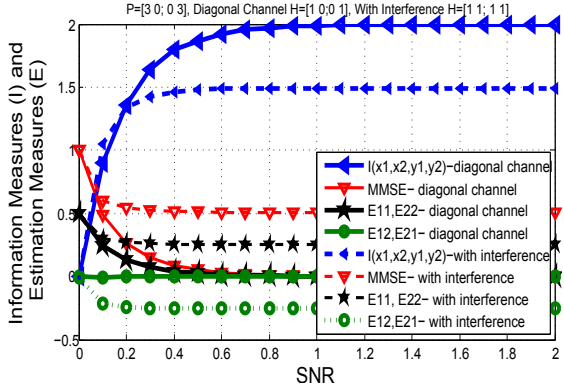}}  
				\caption{Information measures and estimation measures.}
    \label{fig:figure-8.5}
    \end{center}
\end{figure}
\begin{table}[ht]
\caption{Mutual information with and without interference}   % title of Table
\centering % used for centering table
\begin{tabular}{| c | c | c | c | c |}  % centered columns (5  columns)
 \hline\hline
  Signaling & MI without & MI with    & Losses & MIMO setup\\
            &  Int.(bits) & Int.(bits) & (bits) &\\
  \hline
  ${BPSK}$ & $2$  & $1.5$ & $0.5$ & $2\times2$\\
  \hline
  ${QPSK}$ & $4$  & $3$ & $1$ & $2\times2$\\
   \hline
  ${BPSK}$ & $3$  & $1.8$ & $1.2$ & $3\times3$\\
  \hline
  ${QPSK}$ & $6$  & $3.623$ & $2.377$ & $3\times3$\\
  \hline
	${BPSK}$ & $4$  & $2$ & $2$ & $4\times4$\\
  \hline
  ${QPSK}$ & $8$  & $4$ & $4$ & $4\times4$\\
  \hline
%\end{tabular*}
\end{tabular}
\label{table2}
\end{table}
%\vspace{-0.2cm}

Table 2 presents few quantified results for the mutual information with and without interference for different cooperation levels, i.e., for different MIMO setups. The achievable rates and losses are quantified via Monte-Carlo method, for higher constellations like $16-$QAM the number of permutations for: a 2~x~2 MIMO setup are 256, 3~x~3 MIMO setup are 4096, 4~x~4 MIMO setup are 65536; therefore, we limit the presentation to results that are computationally less demanding. Finally, it is worth to present a result that confirms the value of cooperation; i.e., the importance of the MCP network MIMO. 

Figure~\ref{fig:figure-9.5} not only illustrates the value of precoding in comparison to power allocation techniques, but it shows also that even for a diagonal channel without cooperation, i.e., without interference, a non-diagonal precdoing matrix $\bf{P^{\star}}$ is a rate maximizer and better than the mercury waterfilling power allocation $\bf{P_{TPC}}$ with a total power constraint alone and minimum distance of $\sqrt{6}$. In addition, the precoding which inherently includes power allocation and minimum distance maximization is also better than the power allocation with per user power constraint $\bf{P_{UTPC}}$ with a minimum distance of $\sqrt{8}$, for both the precoder and power allocation. 

The matrices used for comparison are as follows: $\bf{H}$=[$\sqrt{3}$ 0; 0 1], $\bf{P}_{TPC}$=[1/$\sqrt{2}$~~0;~0 ~~$\sqrt{3/2}$], $\bf{P}_{UTPC}$=[1 0; 0 1], and $\bf{P^{\star}}$=[1/$\sqrt{2}$~~1/$\sqrt{2}$;-1/$\sqrt{2}$~~1/$\sqrt{2}$]. Therefore, from a precoding perspective, the result illustrates a new look into interference through which a studied one can be a rate maximizer. Moreover, this casts insights that from a network level perspective, coding across packets could be of particular relevance to achieve the network capacity. Such framework introduces network coding over coefficients drawn from smaller sets; called GFs, or a special case of analog network coding. %NC is the focus of the last two Chapters of this PhD thesis, through which we addressed the delay problem with scenarios and technologies in practice.

\begin{figure}[ht!]
%\vspace{0.5cm}
    \begin{center}
        \mbox{\includegraphics[height=2in,width=2.9in]{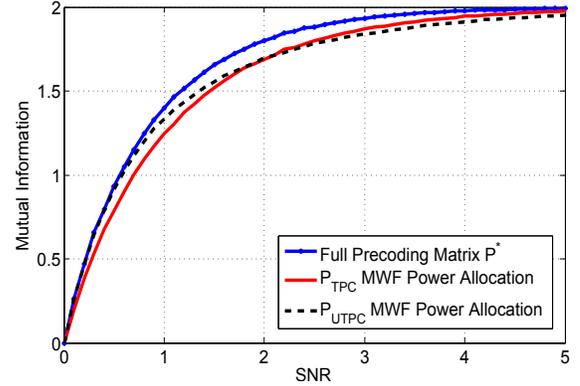}}  
        \caption{Mutual information for BPSK inputs with precoding and power allocation.}
    \label{fig:figure-9.5}
    \end{center}
%		\vspace{0.5cm}
\end{figure}

\section{Conclusion}
In this paper, we studied a cooperative framework where multi-cell processing is used between a cluster of two base stations. We derived the optimal power allocation and the optimal precoding structure which have been found to constitute the optimal setups for MIMO channels. We generalize a non-unique fixed point equation for the optimal precoding and power allocation in the two asymptotic regimes of the high- and low-snr. We provide an iterative approach for the design of the optimal precoding matrix for BPSK constellations at the high-snr. We build upon two distributed algorithms for the optimal solutions in the uplink and the downlink. It has been shown that the cooperation introduces a new look to interference through which a studied addition of interference can add positively to the spectral efficiency of the network. We have also highlighted the coupling between the information rates and the error rates through which the error - or particularly the covariance - caused by the interference links substitutes the drop in the information rates in the main links, this casts insights into having an interpretation of the interference with respect to the channel, transmitted power, and the error. In addition, this has explained why a non-studied interference is a capacity limiting factor in communication channels. Besides the implications of our designs, on defining fundamental limits of cooperation and providing new optimal designs that mitigate the effect of possible and known interferers. The impact of our studied framework extends to more generalized models that include in its structure more information about the system. In particular, the framework introduced in this paper casts insights into investigating the connections between information-theoretic measures and estimation-theoretic measures on a network level~\cite{samahWoWMoM16}, \cite{networkIMMSE}, \cite{piggybacking}, \cite{precodingNetworkLevel}, . Therefore, the system model can include the geometrical properties of the nodes in the network, and this structure exploits a network coding framework. In such framework, it is instrumental to revisit our derivations and design optimal setups that are adapted to the network in conjunction with the knowledge about the physical system. Moreover, of particular relevance are the implications of our derivations and optimal designs on other applications of measurement systems. For instance, systems that are not only interested in reconstructing the original data with lowest error rates, from an estimation perspective, but also aims to do a classification of the data into certain classes. More specifically, the optimal percoding matrix fixed point equation performs a pre-processing over the original data before it is contaminated by noise, and so, it acquires in its structure a maximization of the data rates or information obtained. Similarly, such structure provides studied projections that can be of importance to validate in reconstructing the signal from a compressive measurement. If the optimal precoder for arbitrary inputs distributions is a one which has a non-diagonal structure and a minimum distance maximizer, we could expect that the sparsity of compressed measurements can be designed with similar setups in order to be reconstructed correctly. Furthermore, our formulation of an iterative solution to solve the NP-hard problem can be used to solve similar problems, minimizing the search space into smaller dimensions, via smaller search spaces that are practically relevant to the physical systems under study.

%\newpage
\appendices
%\section{Appendix}
%\begin{center}
%\begin{minipage}[c]{\textwidth}
\section{Proof of Theorem 2} 
Theorem~\ref{theorem2.5} follows the KKT conditions solving~\eqref{eq.a8.5}, subject to~\eqref{eq.a9.5} the relation between the gradient of the mutual information and the MMSE. First, we define the Lagrangian of the optimization problem as follows:
\begin{equation}
\mathcal{L}(P_{1},P_{2},\lambda_1,\lambda_2)=-I(x_{1},x_{2};y_{1},y_{2})-\lambda_1(Q_1-P_1)-\lambda_2(Q_2-P_2)-\mu_1P_1-\mu_2P_2         
\end{equation}
The relation between the gradient of the mutual information with respect to the diagonal power allocation matrix ${\bf{P}}=diag\left({\sqrt{P_1}},{\sqrt{P_2}}\right)$ and the MMSE for linear vector Gaussian channels (MIMO) is given by:
\begin{equation}
\nabla_{P} I(x_{1},x_{2};y_{1},y_{2})=\bf{H}^{\dag}\bf{H}\bf{P}\bf{E}
\end{equation}

and,
\begin{equation}
\nabla_{P{P}^{\dag}} I(x_{1},x_{2};y_{1},y_{2})\bf{P}=\bf{H}^{\dag}\bf{H}\bf{P}\bf{E}
\end{equation}

Given that the inputs covariance and the noise covariance are identities. 
\vspace{0.1cm}
To define the conditions of the theorem, lets re-write the gradient of the Lagrangian:
\vspace{.05cm}
\begin{equation}
\nabla_{P_1}\mathcal{L}(P_{1},P_{2},\lambda_1,\lambda_2)=-\nabla_{P_1}I(x_{1},x_{2};y_{1},y_{2})+\lambda_1-\mu_1, 
\end{equation}
and,
\begin{equation}
\nabla_{P_2}\mathcal{L}(P_{1},P_{2},\lambda_1,\lambda_2)=-\nabla_{P_2}I(x_{1},x_{2};y_{1},y_{2})+\lambda_2-\mu_2,
\end{equation}

with primal feasibility condition, $\lambda_1(Q_1-P_1)=0$, $\mu_1P_1=0$, $\lambda_2(Q_2-P_2)=0$, and $\mu_2P_2=0$, and dual feasibility condition, $\lambda_1\geq0$, and $\lambda_2\geq0$. It follows that:
\begin{multline}
{\lambda_1}^{\star}\sqrt{P_1}=({h_{11}}^{*}h_{11}+{h_{21}}^{*}h_{21})\sqrt{P_1}E_{11}\\
+({h_{11}}^{*}h_{12}+{h_{21}}^{*}h_{22})\sqrt{P_2}E_{21}   
\end{multline}

\vspace{-0.6cm}
\begin{multline}
{\lambda_2}^{\star}\sqrt{P_2}=({h_{12}}^{*}h_{11}+{h_{22}}^{*}h_{21})\sqrt{P_1}E_{12}\\
+({h_{12}}^{*}h_{12}+{h_{22}}^{*}h_{22})\sqrt{P_2}E_{22}   
\end{multline}

{\textbf{Case 1:}} ${P_{1}=0}$, and ${P_{2}>0}$. It follows that:\\
$\mu_1\geq0$, and $\mu_2=0$, taking the gradient with respect to ${P_{1}}$ for the Lagrangian and applying the KKT conditions follows that: ${P_{2}=Q_{2}}$ when $\lambda_2\leq{(h_{12}^{*}h_{12}+ h_{22}^{*}h_{22})E_{22}}$.\\

{\textbf{Case 2:}} ${P_{1}>0}$, and ${P_{2}=0}$. It follows that:\\
$\mu_1=0$, and $\mu_2\geq0$, taking the gradient with respect to ${P_{2}}$ for the Lagrangian and applying the KKT conditions follows that: ${P_{1}=Q_{1}}$ when $\lambda_1\leq{(h_{11}^{*}h_{11}+ h_{21}^{*}h_{21})E_{11}}$.\\

{\textbf{Case 3:}} ${P_{1}>0}$, and ${P_{2}>0}$. It follows that: \\
$\mu_1=0$, and $\mu_2=0$, and the generalized power allocation for both UTs follows:
\begin{multline}
\sqrt{P_1}=\frac{1}{\lambda_1^{\star}}({h_{11}}^{*}h_{11}+{h_{21}}^{*}h_{21})\sqrt{P_1}E_{11}\\
+\frac{1}{\lambda_1^{\star}}({h_{11}}^{*}h_{12}+{h_{21}}^{*}h_{22})\sqrt{P_2}E_{21}   
\end{multline}

\vspace{-0.5cm}
\begin{multline}
\sqrt{P_1}=\frac{1}{\lambda_1^{\star}}({h_{11}}^{*}h_{11}\\
+{h_{21}}^{*}h_{21})\sqrt{P_1}\times \\
\mathbb{E}[(x_{1}-\mathbb{E}(x_{1}|y_{1},y_{2}))(x_{1}-\mathbb{E}(x_{1}|y_{1},y_{2}))^{\dag}]\\
+\frac{1}{\lambda_1^{\star}}({h_{11}}^{*}h_{12}+{h_{21}}^{*}h_{22})\sqrt{P_2}\times \\
\mathbb{E}[(x_{2}-\mathbb{E}(x_{2}|y_{1},y_{2}))(x_{1}-\mathbb{E}(x_{1}|y_{1},y_{2}))^{\dag}].
\end{multline}

\vspace{-0.5cm}
\begin{multline}
P_1^{\star}=\frac{1}{snr1|h_{11}^{*}h_{11}+ h_{21}^{*}h_{21}|}~mmse~(snr1|h_{11}^{*}h_{11}+ h_{21}^{*}h_{21}|P_1^{\star})\\
+\frac{1}{snr2|h_{11}^{*}h_{12}+ h_{21}^{*}h_{22}|}~cov~(snr2|h_{11}^{*}h_{12}+ h_{21}^{*}h_{22}|P_2^{\star}).
\end{multline}

and,
\begin{multline}
\sqrt{P_2}=\frac{1}{\lambda2^{*}}({h_{12}}^{*}h_{11}+{h_{22}}^{*}h_{21})\sqrt{P_1}E_{12}\\
+\frac{1}{\lambda_2^{\star}}({h_{12}}^{*}h_{12}+{h_{22}}^{*}h_{22})\sqrt{P_2}E_{22}
\end{multline}

\vspace{-0.5cm}
\begin{multline}
\sqrt{P_2}=\frac{1}{\lambda_2^{\star}}({h_{12}}^{*}h_{11}+{h_{22}}^{*}h_{21})\sqrt{P_1}\times\\
\mathbb{E}[(x_{1}-\mathbb{E}(x_{1}|y_{1},y_{2}))(x_{2}-\mathbb{E}(x_{2}|y_{1},y_{2}))^{\dag}]\\
+\frac{1}{\lambda_2^{\star}}({h_{12}}^{*}h_{12}+{h_{22}}^{*}h_{22})\sqrt{P_2}\times \\
\mathbb{E}[(x_{2}-\mathbb{E}(x_{2}|y_{1},y_{2}))(x_{2}-\mathbb{E}(x_{2}|y_{1},y_{2}))^{\dag}].   
\end{multline}

\vspace{-0.5cm}
\begin{multline}
P_2^{\star}=\frac{1}{snr2|h_{12}^{*}h_{12}+ h_{22}^{*}h_{22}|}~mmse~(snr2|h_{12}^{*}h_{12}+ h_{22}^{*}h_{22}|P_2^{\star})\\
+\frac{1}{snr1|h_{12}^{*}h_{11}+ h_{22}^{*}h_{21}|}~cov~(snr1|h_{12}^{*}h_{11}+ h_{22}^{*}h_{21}|P_1^{\star}).
\end{multline}

Therefore, Theorem~\ref{theorem2.5} has been proved.

%\end{center}

\section{Proof of Theorem 3}
We can show that the optimum precoding matrix for a MIMO setup satisfies the following fixed point equation:
\begin{equation}
\bf{P}^{\star}=\nu^{-1}\bf{H}^{\dag}\bf{H}\bf{P}^{\star}\bf{E}
\end{equation}
\vspace{-0.5cm}
\begin{equation}
~~=\nu^{-1}\bf{H}^{\dag}\bf{H}\bf{P}^{\star}\mathbb{E}[XX^{\dag}-\mathbb{E}[{X|Y}]\mathbb{E}[{X|Y}]^{\dag}]
\end{equation}
\begin{equation}
~~=\nu^{-1}\bf{H}^{\dag}\bf{H}\bf{P}^{\star}\mathbb{E}[C-\hat{C}]=\nu^{-1}\bf{H}^{\dag}\bf{H}\bf{P}^{\star}\mathbb{E}[U_{C}^{\dag}\Lambda\Pi U_{\hat{C}}],
\end{equation}

via Wiendlant-Hoffman theorem~\cite{102}. With $\nu=||\bf{H}\bf{H}^{\dag}\bf{P}^{\star}\bf{E}||$, ${\bf{X}}=[{x_1~x_2}]^{T}$, and ${\bf{Y}}=[{y_1~y_2}]^{T}$. Therefore, digging into the depth of equation~\eqref{a18}, we can do a singular value decomposition of the channel matrix  $\bf{H}=\bf{U}_{H}\Lambda_{H}\bf{V}_H^{\dag}$ and the MMSE matrix $\bf{E}=\bf{U}_{E}\Lambda_{E}\bf{V}_E^{\dag}$, such that the optimal precoder is: $\bf{P}=\bf{U}\bf{D}\bf{R}^{\dag}$, with $\bf{U}=V_H$ corresponds to the channel right singular vectors, ${\bf{D}}=diag\left(\sqrt{P_1},\sqrt{P_2}\right)$, is a power allocation matrix; i.e., corresponds to the mercury/waterfilling~\cite{43}. $\bf{R}=\Pi U_E$ contains in its structure the eigenvectors of the MMSE matrix which can be permuted and/or projected with $\Pi$ based on the correlation of the inputs and their estimates. Therefore, the rotation matrix insures firstly, allocation of power into the strongest channel singular vectors, and secondly, diagonalizes the MMSE matrix to insure un-correlating the error or in other words independence between inputs, see also \cite{98},\cite{3}, and \cite{92}. Note that each row vector of $\bf{P}^{\star}$ corresponds to the optimal precoding weight that each BS should assign to each transmission in the MCP framework. Therefore, Theorem~\ref{theorem3.5} has been proved.

\section{Proof of Theorem 4}
First we will find the low-snr expansion of the MMSE matrix $\bf{E}$ in~\eqref{a21} as ${snr}\rightarrow 0$, then we will prove \eqref{a22}. The low-snr expansion of the conditional probability exponent is as follows:
\begin{align}
&|{\bf{y}}-\sqrt{snr}{\bf{HP}x}|^2 \nonumber\\
&=\left({\bf{y}}-\sqrt{snr}{\bf{HPx}}\right)^{\dag}\left({\bf{y}}-\sqrt{snr}{\bf{HPx}}\right) \nonumber\\
&=|{\bf{y}}|^2-\sqrt{snr}\left({\bf{y}^{\dag}}\bf{HPx}+\left(\bf{y^{\dag}}\bf{HPx}\right)^{\dag}\right)+{snr}|\bf{HPx}|^2 \nonumber\\
&=|{\bf{y}}|^2-2{\sqrt{snr}}\mathcal{R}\left({\bf{y^{\dag}}}{\bf{HPx}}\right)+{snr}|\bf{HPx}|^2
\end{align}
Hence,
\begin{equation}
p_{{\bf{y|x}}}\left({\bf{y|x}}\right)=\frac{1}{\pi^{n_r}}\exp{\left(-|{\bf{y}}-\sqrt{snr}{\bf{HPx}}|^2\right)}
\end{equation}
\begin{equation}
~~=\frac{1}{\pi^{n_r}}\exp{\left(-|{\bf{y}}|^2\right)}\exp{\left(2{\sqrt{snr}}\mathcal{R}\left(\bf{y^{\dag}HPx}\right)-{snr}|\bf{HPx}|^2\right)}
\end{equation}
\vspace{0.1cm}
%However, due to:
%%\renewcommand{\theequation}{5.73}
%\begin{align}
%&\exp{\left(a\bf{\sqrt{snr}}-b\bf{snr}\right)} \nonumber\\
%&=\frac{\exp{\left(a\bf{\sqrt{snr}}\right)}}{\exp{\left(b\bf{snr}\right)}} \nonumber\\
%&=\frac{1+a\bf{\sqrt{snr}}+\mathcal{O}\left(\bf{snr}\right)}{1+\mathcal{O}\left(\bf{snr}\right)} \nonumber\\
%&=1+a\bf{\sqrt{snr}}+\mathcal{O}\left(\bf{snr}\right),
%\end{align}
The low-snr expansion of the conditional probability distribution of the Gaussian noise is defined as:
\small
\begin{equation}
\label{eq.4.5.1}
{p}_{y|x}\left({\bf{y|x}}\right)=\frac{1}{\pi^{n_r}}\exp{\left(-|{\bf{y}}|^2\right)}\left({1+2\sqrt{snr}}\mathcal{R}\left(\bf{y^{\dag} HPx}\right)+\mathcal{O}\left({snr}\right)\right).
\end{equation}
\normalsize
Therefore,
\begin{multline}
\mathbb{E}_y[\mathbb{E}_{x|y}[x|y]\left(\mathbb{E}_{x|y}[x|y]\right)^{\dag}]=\int_{y\in\mathbb{C}^{n_r}}\frac{1}{\pi^{n_r}}\exp{\left(-|y|^2\right)}\\
\quad\times\frac{snr\left(\bf{HP}\right)^{\dag} yy^{\dag} {\bf{HP}}+\mathcal{O}\left(snr^2\right)}{1+\mathcal{O}\left(snr\right)}dy\\
\end{multline}
It follows that:
\begin{equation}
\label{sectermmmse}
\mathbb{E}_y[\mathbb{E}_{x|y}[x|y]\left(\mathbb{E}_{x|y}[x|y]\right)^{\dag}]=\left({\bf{HP}}\right)^{\dag} {\bf{HP}}snr+\mathcal{O}\left({snr}^2\right)
\end{equation}

The first term of the MMSE matrix $\bf{E}$ is $\mathbb{E}[\bf{xx}^{\dag}]=I$. However, the second term of $\bf{E}$ in \eqref{sectermmmse} is derived by substitution of \eqref{eq.4.5.1} into the conditional mean estimator terms.  Consequently, the low-snr expansion of the MMSE matrix $\bf{E}$ is given as follows:
\begin{equation}
{\bf{E}}={\bf{I}}-({\bf{HP}})^{\dag}{\bf{HP}}.snr+\mathcal{O}(snr^{2}).
\end{equation}

Therefore, we can express the MMSE in terms of the $snr$ as follows:
\begin{align}
&MMSE\left(snr\right) \nonumber\\
&=Tr\left\{\bf{HPE}\left(\bf{HP}\right)^{\dag}\right\} \nonumber\\
&=Tr\left\{{\bf{HP}}\left({\bf{I}}-\left({\bf{HP}}\right)^{\dag} {\bf{HP}}snr+\mathcal{O}\left(snr^2\right)\right)\left({\bf{HP}}\right)^{\dag}\right\} \nonumber\\
&=Tr\left\{{\bf{HP}}\left({\bf{HP}}\right)^{\dag}\right\}-Tr\left\{\left({\bf{HP}}\left({\bf{HP}}\right)^{\dag}\right)^2\right \} snr+\mathcal{O}\left(snr^2\right)
\end{align}

Therefore, Theorem~\ref{theorem4.5} has been proved.

\section{The High-SNR Derivations for BPSK}
%\lipsum[1-2]
The non-linear MMSE matrix $\bf{E}$ is defined as:
\begin{multline}
\label{5.5.0}
\bf{E}=\bf{\mathbb{E}[(\bf{x}-\bf{\mathbb{E}[x|y]})(\bf{x}-\bf{\mathbb{E}[x|y]})^{\dag}]}\\
=\mathbb{E}[\bf{xx^{\dag}}]-\mathbb{E}\bf{\mathbb{E}[x|y]\mathbb{E}[x|y]},
\end{multline}

with,
\begin{equation}
{\bf{\mathbb{E}[x|y]}}=\frac{\sum_{\bf{x}}{\bf{x}}p_{y|x}({\bf{y|x}})p_{x}({\bf{x}})}{p_{y}({\bf{y}})}
\end{equation}

\begin{equation} 
~~= \frac{\sum_{{\bf{x}}}{\bf{x}}p_{y|x}({\bf{y|x}})p_{x}({\bf{x}})}{\sum_{{\bf{x'}}}p_{{y|x'}}({\bf{y|x'}})p_{x}({\bf{x'}})}
\end{equation}

%\vspace{-0.5cm}
For the BPSK inputs, the values of ${\bf{x}}=\{1,-1\}$. Therefore, the non-linear estimate with respect to all possible permutations of the possible inputs is as follows:
\begin{equation} 
{\bf{\mathbb{E}[x|y]}}= \frac{e^{-({\bf{y}}-\sqrt{snr})^{2}}-e^{-({\bf{y}}+\sqrt{snr})^{2}}}{e^{-({\bf{y}}-\sqrt{snr})^{2}}+e^{-({\bf{y}}+\sqrt{snr})^{2}}}
\end{equation}
%
%However,
%%\renewcommand{\theequation}{5.87}
%\begin{equation}
%\mathbb{E}\left[\bf{\mathbb{E}(x|y)}\bf{\mathbb{E}(x|y)^{\dag}}\right]=\int{\left(\frac{\sum_{\bf{x}}\bf{x}p_{y|x}(\bf{y|x})p_{x}(\bf{x})}{p_{y}(\bf{y})}\right)^{2}p_{y}(\bf{y}) \bf{dy}}
%\end{equation}
%
%\vspace{-0.5cm}
%%\renewcommand{\theequation}{5.88}
%\begin{equation}
%\mathbb{E}\left[\bf{\mathbb{E}(x|y)}\bf{\mathbb{E}(x|y)^{\dag}}\right]=\int{\frac{\left(\sum_{\bf{x}}\bf{x}p_{y|x}(\bf{y|x})p_{x}(\bf{x})\right)^{2}}{p_{y}(\bf{y})} \bf{dy}}
%\end{equation}
%
%\vspace{-0.5cm}
%%\renewcommand{\theequation}{5.89}
%\begin{equation}
%\mathbb{E}\left[\bf{\mathbb{E}(x|y)}\bf{\mathbb{E}(x|y)^{\dag}}\right]=\int{\frac{\sum_{\bf{x}}\bf{x}p_{y|x}(\bf{y|x})p_{x}(\bf{x})}{p_{y}(\bf{y})}\sum_{\bf{x}}\bf{x}p_{y|x}(\bf{y|x})p_{x}(\bf{x}) \bf{dy}}
%\end{equation}
%
%\vspace{-0.5cm}
%
%Therefore,
%%\renewcommand{\theequation}{5.90}
%\begin{equation}
%\label{5.5.1}
%\mathbb{E}\left[\bf{\mathbb{E}(x|y)}\bf{\mathbb{E}(x|y)^{\dag}}\right]=\\
%\frac{1}{2\pi}\int{\frac{e^{-(\bf{y-\sqrt{snr}})^{2}}-e^{-(\bf{y+\sqrt{snr}})^{2}}}{e^{-(\bf{y-\sqrt{snr}})^{2}}+e^{-(\bf{y+\sqrt{snr}})^{2}}}\left(e^{-(\bf{y-\sqrt{snr}})^{2}}-e^{-(\bf{y+\sqrt{snr}})^{2}}\right)\bf{dy}}
%\end{equation}
%
%Digging into the depth of the right hand side of equation~\eqref{5.5.1}, we have:
%\renewcommand{\theequation}{5.91}
%\begin{equation}
%(\bf{y-\sqrt{snr}})^{2}=\bf{y^2}-\bf{2\sqrt{snr}y}+\bf{snr}
%\end{equation}
%
%%\renewcommand{\theequation}{5.92}
%\begin{equation}
%(\bf{y+\sqrt{snr}})^{2}=\bf{y^2}++\bf{2\sqrt{snr}y}+\bf{snr}
%\end{equation}
%Thus,
%\renewcommand{\theequation}{5.93}
With,
\begin{equation}
\frac{e^{-({\bf{y}}-\sqrt{snr})^{2}}-e^{-({\bf{y}}+\sqrt{snr})^{2}}}{e^{-({\bf{y}}-\sqrt{snr})^{2}}+e^{-({\bf{y}}+\sqrt{snr})^{2}}}
=\frac{e^{2\sqrt{snr}{\bf{y}}}-e^{-2\sqrt{snr}{\bf{y}}}}{e^{2\sqrt{snr}{\bf{y}}}+e^{-2\sqrt{snr}{\bf{y}}}}
\end{equation}
\begin{equation}
~~=tanh\left(2\sqrt{{snr}}\mathcal{R}(\bf{y})\right)
\end{equation}
%It follows that:
%%\renewcommand{\theequation}{5.95}
%\begin{equation}
%\mathbb{E}\left[\bf{\mathbb{E}(x|y)}\bf{\mathbb{E}(x|y)^{\dag}}\right]=\\
%\frac{1}{2\pi}\int{tanh\left(2\sqrt{\bf{snr}}\mathcal{R}(\bf{y})\right)\left(e^{-(\bf{y-\sqrt{snr}})^{2}}-e^{-(\bf{y+\sqrt{snr}})^{2}}\right)\bf{dy}}
%\end{equation}
%
%Therefore,
%%\renewcommand{\theequation}{5.96}
%\begin{multline}
%\mathbb{E}\left[\bf{\mathbb{E}(x|y)}\bf{\mathbb{E}(x|y)^{\dag}}\right]=\frac{1}{2\pi}\int_{y\in \textsl{C}}{tanh\left(2\sqrt{\bf{snr}}\bf{\mathcal{R}(y)}\right)e^{-\left(\bf{y}-\bf{\sqrt{snr}}\right)^{2}}}\\
%-\frac{1}{2\pi}\int_{y\in \textsl{C}}{tanh\left(2\sqrt{\bf{snr}}\bf{\mathcal{R}(y)}\right)e^{-\left(\bf{y}+\bf{\sqrt{snr}}\right)^{2}}\bf{dy}}
%\end{multline}
%
%However, it is known that:
%\renewcommand{\theequation}{5.97}
%\begin{equation}
%\bf{tanh(-x)}=\bf{-tanh(x)}
%\end{equation}
Since $tanh(-x)=-tanh(x)$ and the expectation remains the same if $\bf{y}\sim$ $N$ $(\sqrt{snr},1)$ replaced by $\bf{y}\sim$ $N$ $(-\sqrt{snr},1)$, due to symmetry, we have:
\begin{equation}
\label{5.5.3}
\mathbb{E}\left[\mathbb{E}({\bf{x|y}}){\mathbb{E}({\bf{x|y}})^{\dag}}\right]=\frac{1}{\pi}\int_{y\in \textsl{C}}{tanh \left(2\sqrt{snr}{\bf{y}}\right)e^{-\left({\bf{y}}-\sqrt{snr}\right)^{2}}{d{\bf{y}}}}
\end{equation}
 %$\bf{\zeta=\sqrt{snr}-2y}$
Therefore, due to marginalization of the complex domain into the real domain, substituting $\bf{\zeta}$ into \eqref{5.5.3}, and $\mathbb{E}[\bf{xx^{\dag}}]=1$ into \eqref{5.5.0}, the mmse(snr) of a BPSK input over a SISO channel is given by:
\begin{equation}
\label{5.5.4}
{mmse(snr)}=1-\frac{1}{\sqrt{\pi}}\int_{\zeta \in \textsl{R}}{tanh\left(2\sqrt{{snr}}{\bf{\zeta}}\right)e^{-\left({\bf{\zeta}}-\sqrt{snr}\right)^{2}} d\bf{\zeta}}
\end{equation}
Due to the fundamental relation between mmse(snr) and the mutual information, we will integrate \eqref{5.5.4} with respect to the snr to get the close form expression of the mutual information $I(snr)$ of a BPSK input over a SISO channel, as follows:
\begin{equation}
{I(snr)}={snr}-\frac{1}{\sqrt{\pi}}\int_{\zeta \in \textsl{R}}log~cosh\left(2\sqrt{{snr}}{\zeta}\right)e^{-\left({\zeta-\sqrt{snr}}\right)^{2}} d{\zeta}
\end{equation}

Therefore, we proved the relation for BPSK inputs as in \eqref{m1}, and \eqref{m2}.

\section*{Acknowledgment}
The author would like to thank Alberto Gil Ramos for helping in some proofs.%, particularly in the expansions and matrix differentiation.
\ifCLASSOPTIONcaptionsoff
  \newpage
\fi

\bibliographystyle{IEEEtran}
\bibliography{IEEEabrv,mybibfile}

\end{document}